\documentclass[runningheads,envcountsame]{llncs}

\usepackage{amsmath} 
\usepackage{amssymb,euscript}
\usepackage{enumerate}
\usepackage{pgf,tikz,pgfplots}
\usepackage{mathrsfs}
\usetikzlibrary{arrows}\usetikzlibrary{arrows.meta}
\usetikzlibrary[patterns]

\usepackage{hyperref}

\def\apxmark{$\!\!${\bf(*)}}

\def\cwp{\mathop{cw}}
\def\ca#1{{\cal#1}}
\def\cf#1{{\EuScript#1}}

\def\ordtype{\mathop{\mbox{\sl ordccw}}}

\def\prebox#1{\mathop{\mbox{\sl #1}}}

\begin{document}

\title{Clique-Width of Point Configurations%
 \thanks{O.~\c{C}a\u{g}{\i}r{\i}c{\i}, P.~Hlin\v{e}n\'y and F.~Pokr\'yvka 
	have been supported by the {Czech Science Foundation}, project no.~{20-04567S}.
	A.~Sankaran has been supported by the Leverhulme Trust through a
	Research Project Grant on `Logical Fractals'.}
}

\author{Onur \c{C}a\u{g}{\i}r{\i}c{\i}\inst1\,\orcidID{0000-0002-4785-7496}
	\and
	Petr~Hlin\v en\'y\inst1\,\orcidID{0000-0003-2125-1514}
	\and
	Filip~Pokr\'yvka\inst1
	\and
	Abhisekh Sankaran\inst2
}

\authorrunning{O.~\c{C}a\u{g}{\i}r{\i}c{\i}, P.~Hlin\v{e}n\'y
	F.~Pokr\'yvka and A.~Sankaran}

\institute{Faculty of Informatics, Masaryk University, Brno, Czech Republic
	\\\email{onur@mail.muni.cz}\quad\email{hlineny@fi.muni.cz}\quad
		\email{xpokryvk@fi.muni.cz}
	\medskip\and
	Department of Computer Science and Technology, University of
	Cambridge, UK \\\email{abhisekh.sankaran@cl.cam.ac.uk}}

\maketitle

\begin{abstract}
While structural width parameters (of the input) belong 
to the standard toolbox of graph algorithms,
it is not the usual case in computational geometry.
As a case study we propose a natural extension of the structural graph
parameter of {\em clique-width} to geometric point configurations
represented by their {\em order type}.
We study basic properties of this clique-width notion,
and relate it to the monadic second-order logic of point configurations.
As an application, we provide several linear FPT time algorithms for
geometric point problems which are NP-hard in general,
in the special case that the input point set is of bounded clique-width
and the clique-width expression is also given.

\end{abstract}

\section{Introduction}\label{introduction}

An order type is a useful means to characterize the combinatorial properties
of a finite point configuration in the plane.
As introduced in Goodman and
Pollack~\cite{DBLP:journals/siamcomp/GoodmanP83,DBLP:journals/dcg/GoodmanP86},
the {\em order type} of a given set $P$ of points assigns, to each ordered triple
$(a,b,c)\in P^3$ of points, the orientation (either clockwise or
counter-clockwise) of the triangle $abc$ in the plane.
More generally, if the point set $P$ is not in a~general position, the triple
$(a,b,c)$ may also be collinear (as the natural third option).

Knowing the order type of a point set $P$ is sufficient to determine some
useful combinatorial properties of the geometric set $P$, such as
the convex hull of~$P$ and other.
For example, problems of finding convex holes in $P$ or dealing with the
intersection pattern of straight line segments with ends in~$P$,
can be solved by looking only at the order type of $P$ and not on its
geometric properties.
That is why order types of points sets are commonly studied
from various perspectives in the field of computational geometry, e.g., 
\cite{DBLP:conf/stoc/GoodmanPS89,DBLP:journals/order/AichholzerAK02,%
DBLP:journals/comgeo/AichholzerK07,%
DBLP:conf/soda/AloupisILOW14,DBLP:conf/compgeom/GoaocHVSV15,%
DBLP:journals/ijcga/Roy16,%
DBLP:journals/ijcga/AichholzerKMPW17,DBLP:conf/gd/AichholzerB0KMP19}.

On the other hand, knowing the order type of $P$ is obviously not sufficient to
answer questions involving ``truly geometric'' aspects of $P$, e.g., 
distances in~$P$ (straight-line or geodesic), or 
angles between the lines or the area of polygons within~$P$.
Nevertheless, even in such geometry-based problems, a more efficient
subroutine computing with the order type of $P$ might speed-up the overall
computation, which can be a promising direction for future research.

Unlike in the area of graphs and graph algorithms, where structural 
width parameters are very common for many years, at least since the 90's, 
no similar effort can be seen in combinatorial and computational geometry.
We would like to introduce, in this paper, possible combinatorial handling 
of ``structural complexity'' of a given point configuration~$P$
through defining its ``width'' (which we would assume to be small for the
studied inputs).
Note that, although the desired width would be a concrete natural number, we
will not be interested in the exact value of it, but instead study whether
the width would be bounded or unbounded on a given class of point
configurations.

Inspired by graph structure parameters, the obvious first attempt could be
to extend the traditional notion of {\em tree-width}~\cite{gmii}.
Such an extension is technically possible (cf.~tree-width of the Gaifman graph
of a relational structure), but the huge problem is that for the tree-width
to be upper-bounded, the underlying structure must be ``sparse'' -- in
particular, it can only have a linear number of edges\,/\,tuples.
This is clearly not satisfied for the order type in which
about half of all triples are of each orientation.

A better option comes with another traditional, but not so well-known,
notion of {\em clique-width}~\cite{CourcelleER1991}.
Clique-width can be bounded even on dense graphs, such as on cliques,
and, similarly to the case of Courcelle's theorem~\cite{cou90} for tree-width, 
clique-width also enjoys some nice metaalgorithmic properties,
e.g.~\cite{cmr00,DBLP:journals/dam/GanianH10}.
This includes solving any decision (and some optimization as well) problems
formulated in the monadic second-order (MSO) logic in linear time.
Hence, alongside with the (Section~\ref{sec:ordertypes}) 
proposed definition of the clique-width of point configurations, 
we will introduce the MSO language of their order types
and discuss which problems can be formulated in this language
(and hence solved in linear time if a point set with a
decomposition of bounded clique-width is given on the input).

\vspace*{-3ex}%
\subsubsection{Paper organization. }
We introduce order types of point configurations, viewed as ternary
relational structures, in Section~\ref{sec:ordertypes}.
Then we formally define their clique-width as that of relational structures.
We also show why a technically simpler-looking ``unary'' clique-width
(which is closer to the traditional graph clique-width) does not
work well for order types.

In Section~\ref{sec:MSO} we speak about MSO logic of order types, and give
a basic overview of its use and expressive power.
We restate classical logic results on clique-width characterization
and metaalgorithmics (Theorems~\ref{thm:cwtransdu} and~\ref{thm:MSOlin}).
We continue the study in Section~\ref{sec:assorted} with a few concrete
interesting examples of bounding the clique-width of special point sets
(Theorem~\ref{thm:cwbasic2}),
and of solving some geometric point problems, which are otherwise NP-hard,
for inputs of bounded clique-width
(Theorems~\ref{thm:genposs}, \ref{thm:convparti} and~\ref{thm:terraing}).

We conclude with some questions and suggestions 
in Section~\ref{sec:conclu}.
Due to space restrictions, details of the ~\apxmark-marked statements are
left for the Appendix.

\section{Order Types and Clique-Width}
\label{sec:ordertypes}

We now recall the notion of an order type in a formal setting,
and propose a definition of the clique-width of (the order type~of) 
a point configuration, based on a natural specialization of 
the very general concept of clique-width of relational structures.
A {\em relational structure} $S=(U,\>R_1^{S},\dots,R_a^{S})$
of the signature $\sigma=\{R_1,\dots,R_a\}$ consists of a
universe (a finite set) $U$ and a (finite) list of relations
$R_1^{S},\dots,R_a^{S}$ over~$U$.
For instance, for graphs, $U=V(G)$ is the vertex set and $R_1^{G}=E(G)$ 
is the binary symmetric relation of edges of~$G$.

For a set of points $P$, here {\em always} considered in the plane, consider a map
$\omega:P^3\to\{+,-,0\}$ where $\omega(a,b,c)=0$ if the triple of points%
\footnote{Note that if any two of $a,b,c$ are not distinct, then
we automatically get $\omega(a,b,c)=0$, and so when we then shift to seeing an
order type as a ternary relation, the involved triples would always consist
of distinct elements (which is technically nice).}
$a,b,c$ is collinear, $\omega(a,b,c)=+$ if $abc$ forms a counter-clockwise
oriented triangle, and $\omega(a,b,c)=-$ otherwise.
Then $\omega$ is traditionally called the {order type} of $P$, but we,
for technical reasons, prefer defining the {\em order type} of $P$ as the
ternary relation $\Omega\subseteq P^3$ such that 
$(a,b,c)\in\Omega$ iff $\omega(a,b,c)=+$.
Hence we have formally got a relational structure $(P,\Omega)$ of the
signature consisting of one ternary symbol.
We will also write $\Omega(P)$ to emphasize that $\Omega$ is the
order type of the point set~$P$.

Observe that $\omega(a,b,c)=\!-\,$ iff $(b,a,c)\in\Omega$,
and $\omega(a,b,c)=0$ iff $(a,b,c),(b,a,c)\not\in\Omega$. Hence,
the relation $\Omega$ fully determines the usual order type of~$P$. 
Furthermore, whenever $(a,b,c)\in\Omega$, we also have
$(b,c,a)\in\Omega$ and $(c,a,b)\in\Omega$,
and so we call the set of triples $\{(a,b,c),\, (b,c,a),\, (c,a,b)\}$
the {\em cyclic closure} of $(a,b,c)\in\Omega$.

\paragraph{Unary clique-width. }
We start with the definition of ordinary graph clique-width.
Let a {$\ell$-expression} be an algebraic expression
using the following four operations on vertex-labelled graphs using $\ell$ labels:
\begin{enumerate}[(u1)]\vspace*{-1ex}
\item create a new vertex with single label $i$;
\item\label{it:uunion} take the disjoint union of two labelled graphs;
\item\label{it:labeledge}
add all edges between the vertices of label $i$ and label $j$ ($i\not=j$); and
\item relabel all vertices with label $i$ to label $j$.
\end{enumerate}\vspace*{-1ex} 
The {\em clique-width} ${\rm cw}(G)$ of a graph $G$ equals the minimum $\ell$ such that
(some labelling of) $G$ is the value of an $\ell$-expression.

The idea behind this definition is that the edge set of a graph $G$ can be constructed
with ``bounded amount of information''; this is since we have only a fixed
number of distinct labels and vertices of the same label are further
indistinguishable by the expression.

This definition has an immediate generalization to the {\em unary clique-width}
of an order type $\Omega(P)$ of a point set $P$ (the adjective referring to the 
fact that labels occur as unary predicates in the definition):
replace (u\ref{it:labeledge}) with
\begin{description}\vspace*{-1ex}\item[\rm(u3')]
add to $\Omega$ the cyclic closures of all triples $(a,b,c)$ of distinct
elements such that $a$ is labelled $i$, $b$ is labelled $j$ and $c$ is labelled $k$.
\end{description}\vspace*{-1ex}
Unfortunately, although being very simple, this definition is generally not
satisfactory due to problems discussed, e.g.,
in~\cite{DBLP:journals/corr/abs-0806-0103}
and specifically illustrated for order types in our Proposition~\ref{prop:unaryvs}.

\paragraph{Multi-ary clique-width. }
While in the case of graphs (whose edge relation is binary) it is sufficient
to consider clique-width expressions with unary labels,
for the ternary order-type relation (as well as for other relational
structures of higher arity) it is generally necessary 
to allow creation of ``intermediate'' binary labels, which are labelled
pairs of points of $P$.

This generalization, which is in agreement with with the treatment by 
Blumensath and Courcelle~\cite{DBLP:journals/iandc/BlumensathC06},
leads to the proposed new definition:

\begin{definition}[Clique-width of a point configuration] 
\label{def:cwpoints}\rm
Consider an algebraic expression $\cal E$ using the following five operations 
on labelled relational structures (of arity $3$ in this case) over point sets:
\begin{enumerate}[(w1)]\vspace*{-1ex}
\item\label{it:wcreate} create a new point with single label $i$;
\item\label{it:wdisju} take the disjoint union of two point sets;
\item\label{it:wpairl} for every two points, point $a$ of label $i$ and point $b$ of label $j$
($i\not=j$), give the ordered pair $(a,b)$ binary label $k$;%
\footnote{After this operation, $(a,b)$ may hold more than one binary label, which is ok.}
\item\label{it:wtriple} for every three pairwise distinct points,
$a$, $b$ and $c$ such that $c$ is of (unary) label $i$,
and the pair $(a,b)$ is of (binary) label $k$,
add to the structure the cyclic closure of the ordered triple $(a,b,c)$,
\item[(w\ref{it:wtriple}')$\!$]
under the same conditions as in (w\ref{it:wtriple}), add the cyclic closure of $(b,a,c)$,
\item\label{it:wrelab} 
relabel all tuples (singletons or pairs) with label $i$ to label $j$ of equal arity.
\end{enumerate}\vspace*{-1ex} 
The {\em value} of such expression $\cal E$ is the ternary relational
structure on the points created by (w\ref{it:wcreate}) and consisting of the
triples added by (w\ref{it:wtriple}) and (w\ref{it:wtriple}').
The auxiliary labels introduced in $\cal E$ are no longer relevant after the
evaluation of $\cal E$.

The {\em width} of an expression $\cal E$ constructed as in (w1)--(w5) equals the sum
of arities of the labels occuring in $\cal E$.%
\footnote{Note that this `sum of arities' measure directly generalizes the
number $\ell$ of unary labels in the expression of (u1)--(u4).}
The {\em clique-width} $\cwp(P)$ of a point configuration $P$
equals the minimum $\ell$ such that the order type $\Omega(P)$ of $P$
is the value of an expression of width at most~$\ell$.
\end{definition}

\begin{remark}\label{rem:spacelinear}
Notice that Definition~\ref{def:cwpoints} does not address the question 
of realizability of the relational structure of $\cal E$ as an order type.
This is formally right since we compare the value of
constructed $\cal E$ to the order type of an existing point~set.
\end{remark}

\begin{remark}
There is one side effect of Definition~\ref{def:cwpoints} which also deserves attention.
To give an order type $\Omega(P)$ as a ternary relation, in general, 
one needs cubic space to list the counter-clockwise triples.
On the other hand, assuming bounded clique-width of~$P$, the expressions
evaluating to $\Omega(P)$ is of linear size
which is much smaller (and comparable to listing the point coordinates).
This would admit the possibility of linear-time algorithms with order types.
\end{remark}

For a closer explanation of this concept, we present a basic example:
\begin{proposition}\apxmark\label{prop:unaryvs}
Let $P$ be an arbitrary finite set of points in a strictly convex position.%
\footnote{That is, in the convex hull of $P$ every point of $P$ is a vertex.}
Then the clique-width of $P$ is bounded by a constant, while the unary
clique-width of $P$ is unbounded.
\end{proposition}
\paragraph{Proof outline. }
Let the points of $P$ be $p_1,p_2,\ldots,p_n$ in   
the counter-clockwise order (starting arbitrarily). 
We start with $p_1$ and stepwise add $p_2$, $p_3$ etc.,
changing previous points to label $1$ and the added point created with unique label~$2$.
See Figure~\ref{fig:convex-cw}.
Along the steps, right after the creation of $p_j$, we add the binary label
$3$ to all pairs labelled $1$ and $2$, i.e., to $(p_i,p_j)$ for all $i<j$, 
and create the order triple $(p_i,p_{i'},p_j)$ over all pairs $(p_i,p_{i'})$ 
of label~$3$, $i,i'\not=j$ and $p_j$ of label~$2$.
This construction witnesses that the clique-width of $P$ is at most~$4$.

On the other hand, take unary clique-width with $\ell$ labels, 
and $|P|\geq 2\ell+1$.
An arbitrary $\ell$-expression for $\Omega(P)$ must have a union
operation (the ``last'' one) over two subsets such that one has more than
$\ell$ points, and so two points $a,b$ of the same label by the pigeon-hole
principle.
Let $c$ be any point from the other set.
Then there is no way, based on the labels, to distinguish between the
triples $(a,b,c)$ and $(b,a,c)$, which must have the opposite orientations
in~$\Omega(P)$. Therefore, the clique-with of $P$ must be at least $\ell+1$.

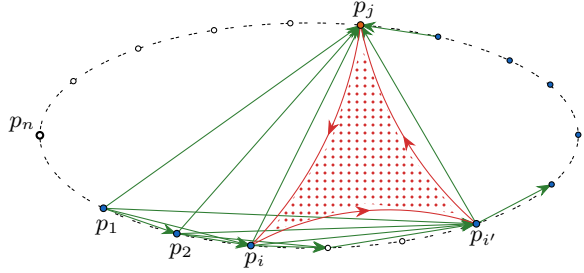
\begin{figure}[t]
\definecolor{dtsfsf}{rgb}{0.8274509803921568,0.1843137254901961,0.1843137254901961}
\definecolor{sexdts}{rgb}{0.1803921568627451,0.49019607843137253,0.19607843137254902}
\definecolor{rvwvcq}{rgb}{0.08235294117647059,0.396078431372549,0.7529411764705882}
\definecolor{dbwrru}{rgb}{0.8588235294117647,0.3803921568627451,0.0784313725490196}
$$
\begin{tikzpicture}[scale=0.5, line cap=round,line join=round, >={Stealth[scale=1.2]} ,x=1cm,y=1cm]
\path [use as bounding box] (-13.5,-3.8) rectangle (15,1.4);
\fill[line width=0pt,dotted,color=dtsfsf,fill=dtsfsf,pattern=dots,pattern color=dtsfsf] (-2.4948287498119,-3.5270618362800907) -- (-0.005408729963423409,-2.807032191163838) -- (2.4278370815092054,-3.0712131649808665) -- (0.9817938564054718,-1.2636591336011993) -- (-0.10273856242232857,1.3364378192295527) -- (-0.9369942692129442,-1.1802335629221377) -- cycle;
\draw [rotate around={0:(-1.5,-1)},line width=0.4pt,dash pattern=on 1pt off 1pt on 1pt off 4pt] (-1.5,-1) ellipse (7.1575008039952195cm and 2.9966344053274447cm);
\draw [->,color=sexdts] (-6.966963666482887,-2.93415464949699) -- (-5.0158311700965905,-3.610193099365358);
\draw [->,color=sexdts] (-5.0158311700965905,-3.610193099365358) -- (-3.0480865026272097,-3.9257024297425205);
\draw [->,color=sexdts] (-5.0158311700965905,-3.610193099365358) -- (2.948509209473488,-3.3475651581673125);
\draw [->,color=sexdts] (-3.0480865026272097,-3.9257024297425205) -- (-1.0002988723446915,-3.9893224793422504);
\draw [->,color=sexdts] (-5.0158311700965905,-3.610193099365358) -- (-1.0002988723446915,-3.9893224793422504);
\draw [->,color=sexdts] (-6.966963666482887,-2.93415464949699) -- (2.948509209473488,-3.3475651581673125);
\draw [->,color=sexdts] (-3.0480865026272097,-3.9257024297425205) -- (2.948509209473488,-3.3475651581673125);
\draw [->,color=sexdts] (-6.966963666482887,-2.93415464949699) -- (-0.13054708598201437,1.9412732066527454);
\draw [->,color=sexdts] (-5.0158311700965905,-3.610193099365358) -- (-0.13054708598201437,1.941273206652745);
\draw [->,color=sexdts] (-3.0480865026272097,-3.9257024297425205) -- (-0.13054708598201437,1.941273206652745);
\draw [->,color=sexdts] (2.948509209473488,-3.3475651581673125) -- (-0.13054708598201437,1.9412732066527454);
\draw [->,color=sexdts] (-6.966963666482887,-2.93415464949699) -- (-3.0480865026272097,-3.9257024297425205);
\draw [->,color=sexdts] (-1.0002988723446915,-3.9893224793422504) -- (2.948509209473488,-3.3475651581673125);
\draw [->,color=sexdts] (2.948509209473488,-3.3475651581673125) -- (4.949471505683016,-2.2995072928454503);
\draw [shift={(-10.34370768766397,3.3611834734331034)},color=dtsfsf]  plot[domain=5.498386168405508:6.1450432904006735,variable=\t]({1*10.311391498796697*cos(\t r)+0*10.311391498796697*sin(\t r)},{0*10.311391498796697*cos(\t r)+1*10.311391498796697*sin(\t r)});
\draw [shift={(0.611346984902906,-10.49411088074886)},color=dtsfsf]  plot[domain=1.2547259984909083:2.079094361962844,variable=\t]({1*7.519005454773145*cos(\t r)+0*7.519005454773145*sin(\t r)},{0*7.519005454773145*cos(\t r)+1*7.519005454773145*sin(\t r)});
\draw [shift={(9.23611770712632,3.8536569932382543)},color=dtsfsf]  plot[domain=3.3429937578163544:3.9946186522884526,variable=\t]({1*9.559896500153666*cos(\t r)+0*9.559896500153666*sin(\t r)},{0*9.559896500153666*cos(\t r)+1*9.559896500153666*sin(\t r)});
\draw [->,color=dtsfsf] (-0.15774902693414428,-3.0145430888584315) -- (0.1351107979547299,-2.9902024242743686);
\draw [->,color=dtsfsf] (1.125468322368782,-1.2068753553573428) -- (1,-1);
\draw [->,color=dtsfsf] (-0.8532798464287286,-0.6707607265483104) -- (-1,-1);
\draw [->,color=sexdts] (1.930362456051306,1.630051633909205) -- (-0.13054708598201414,1.9412732066527454);
\draw [fill=dbwrru] (-0.1305470859820142,1.9412732066527454) circle (2.5pt);
\draw (-0.019312991743267005,2.3402148349193864) node {$p_j$};
\draw [fill=rvwvcq] (-6.966963666482887,-2.93415464949699) circle (2.5pt);
\draw (-6.8602097874263155,-3.440181049805715) node {$p_1$};
\draw [fill=rvwvcq] (-5.0158311700965905,-3.610193099365358) circle (2.5pt);
\draw (-4.8997088764683685,-4.075856152382075) node {$p_2$};
\draw [fill=rvwvcq] (-3.0480865026272097,-3.9257024297425205) circle (2.5pt);
\draw (-2.9392079655104215,-4.3273836361746103) node {$p_i$};
\draw [color=black,fill=white] (-1.0002988723446915,-3.9893224793422504) circle (2pt);
\draw [color=black,fill=white] (0.9769464781854367,-3.811476148864313) circle (2pt);
\draw [fill=rvwvcq] (2.948509209473488,-3.3475651581673125) circle (2.5pt);
\draw (3.1535288616021674,-3.7434046414211795) node {$p_{i'}$};
\draw [color=black,fill=rvwvcq] (4.949471505683016,-2.2995072928454503) circle (2pt);
\draw [color=black,fill=rvwvcq] (5.657321614043455,-0.9787957866810123) circle (2pt);
\draw [color=black,fill=rvwvcq] (4.913169872205855,0.3561873637505762) circle (2pt);
\draw [color=black,fill=rvwvcq] (3.8336077751881827,0.9983544845594516) circle (2pt);
\draw [color=black,fill=rvwvcq] (1.930362456051306,1.630051633909205) circle (2pt);
\draw [color=black,fill=white] (-1.9997011276552858,1.9893224793422513) circle (2pt);
\draw [color=black,fill=white] (-4.026785812379924,1.803691333878541) circle (2pt);
\draw [color=black,fill=white] (-6.043437868702751,1.3154783637099658) circle (2pt);
\draw [color=black,fill=white] (-7.766333987623466,0.44807200858614693) circle (2pt);
\draw [fill=black,fill=white,thick] (-8.657480256192192,-0.9928195701802587) circle (2.5pt);
\draw[color=black] (-9.1262546278739,-0.6935844006276118) node {$p_n$};
\end{tikzpicture}
$$
\caption{An illustration of the expression (width~$4$) in
Proposition~\ref{prop:unaryvs}. Unary labels $1$ are blue (on
$p_1,\ldots,p_{j-1}$), the unique label $2$ is orange (on $p_j$ just added),
and the binary labels $3$ are with green arrows.
We are just creating the red triple(s) $(p_i,p_{i'},p_j)$.}
\label{fig:convex-cw}
\end{figure}

\vspace*{-1ex}%
\paragraph{Annotated point configurations. }
In some situations, it may be useful to consider a point configuration $P$
with additional information (or structure) on the points or selected pairs of them.
An exemplary use case for such annotations is to study polygons, with $P$ as the vertex set,
for which case we are considering an order type $\Omega(P)$ together with a directed Hamiltonian
cycle on $P$ representing the counter-clockwise boundary~of~$P$.

Formally, we simply consider relational structures (over $P$)
with the signature consisting of
the ternary order type and arbitrary binary or unary symbols.
The {\em clique-width} of such an {\em annotated point configuration} $P$
is, naturaly, as in Definition~\ref{def:cwpoints} with additional rules that
some of the auxiliary unary and binary labels are at the end turned into the
desired unary and binary relations on~$P$.


\section{MSO logic of order types}
\label{sec:MSO}

The beginning of this section is devoted to a short introduction of the 
{\em monadic second-order (MSO) logic} of relational structures.
Recall a {\em relational structure} $S=(U,\>R_1^{S},\dots,R_q^{S})$
of the signature $\sigma=\{R_1,\dots,R_q\}$.

The language of MSO logic (of the signature~$\sigma$) then consists
of the standard propositional logic,
quantifiers $\forall,\exists$ ranging over elements and subsets of the
universe $U$,
and the relational symbols $R_1,\dots,R_q$ with the following meaning:
for $R_i$ of arity $a$, we have $S\models R_i(x_1,\ldots,x_a)$
if and only if $(x_1,\ldots,x_a)\in R_i^{S}$.

In our specific case of order types $\Omega(P)$ of point sets~$P$,
we use the relational symbol $\ordtype(x_1,x_2,x_3)$ for $\Omega$ within MSO logic.
For example, we can express that a point $y$ lies strictly in the convex
hull of points $x_1,x_2,x_3$ as follows
\begin{eqnarray}\label{eq:triangle}
 \left[\ordtype(x_1,x_2,x_3)\wedge\ \bigwedge\nolimits_{i=1,2,3}
	\ordtype(x_i,x_{i+1},y)\right] \vee 
\\\nonumber \vee
 \left[\ordtype(x_3,x_2,x_1)\wedge \bigwedge\nolimits_{i=1,2,3}
	\ordtype(x_{i+1},x_i,y)\right]
,\end{eqnarray}
where $x_4$ is taken as~$x_1$.

More generally, we can express that a point $y\in P$ belongs to the convex hull
(not necessarily strictly now) of a set $X\subset P$, $y\not\in X$,
with the following formula:
\begin{eqnarray}\label{eq:convhull}
\prebox{convhull}(X,y) \>\equiv\>
  \forall x,x'\in X && \big[ \big(x\not=x' \wedge 
	\forall z\in X \neg\ordtype(x',x,z) \big)
\\\nonumber
  &&  \to \neg\ordtype(x',x,y) \big]
\end{eqnarray}
Then we may express, for example, that a set $X\subseteq P$ is a convex hole
(i.e., no point outside of $X$ belongs to the convex hull of $X$,
and no point of $X$ belongs to the convex hull of the rest of $X$)
with the following:
\begin{equation}\label{eq:convhole}
 \forall y\!\not\in\!X\, (\neg\prebox{convhull}(X,y)) \>\wedge\>
 \forall Y\!\subseteq\!X\forall z\!\in\!X ( \prebox{convhull}(Y,z) \to z\in Y )
\end{equation}
Further similar examples are easy to come with.

\paragraph{Interpretations and transductions. }
We sketch the concept of ``translating''
between relational structures.
Consider relational signatures $\sigma=\{R_1,\dots,R_q\}$ and
$\tau=\{R'_1,\dots,R'_{t}\}$.
A (simple) {\em MSO interpretation of $\tau$-structures in
$\sigma$-structures} is a $t$-tuple of MSO formulas
$\Psi=(\psi_i: 1\leq i\leq t)$ of the signature $\sigma$,
where the number of free variables of $\psi_i$ equals the arity $a_i$ of $R'_i$.
A $\tau$-structure $T$ is interpreted in a $\sigma$-structure $S$
via $\Psi$ if $T$ and $S$ share the same ground set~$U$ and, 
for each $1\leq i\leq t$, we have $(x_1,\ldots,x_{a_i})\in {R'_i}^{T}$ $\iff$
$S\models \psi_i(x_1,\ldots,x_{a_i})$.

As a short example, consider a point set $P$ and its mirror image $P'$.
Then the order type $\Omega(P')$ can be interpreted in $\Omega(P)$ simply by
taking $\psi_1(a,b,c)\equiv\ordtype(b,a,c)$.
The true power of interpretations will show up in the following.

There is a more general concept of a {\em transduction} from a
$\sigma$-structure $S$ to a set of $\tau$-structures which,
before taking an (MSO) interpretation, 
has abilities (in this order of application); 
(i) to equip $S$ with a fixed number of arbitrary parameters given as unary labels
(because of this, the result of a transduction is not deterministic, but a set of $\tau$-structures),
(ii) to ``amplify'' the ground set of $S$ by taking a bounded number 
of disjoint copies of $S$, and
(iii) to subsequently restrict the ground set by a unary MSO formula.
See the Appendix and/or Courcelle and Engelfriet~\cite{ce12} for more
technical details on transductions.

For a class of relational structures $\cf S$, the image in a transduction $\Psi$ of
the class $\cf S$ is the union of all transduction results,
precisely, $\Psi(\cf S):=\bigcup_{S\in\cf S}\Psi(S)$.
We say that a class $\cf S$ is of {\em bounded clique-width} if there exists
a constant $h$ such that the clique-width of every $S\in\cf S$ is at most~$h$.
Obviously, this is only an asymptotic concepts which makes sense for
infinite classes $\cf S$ and, mainly, it ``smoothens'' marginal technical
differences between various definitions of clique-width.
On this abstract level, we then obtain the following
{\em crucial characterization}:

\begin{theorem}[
	Blumensath~and Courcelle~{\cite[Proposition~27]{DBLP:journals/iandc/BlumensathC06}}]
\,\apxmark\label{thm:cwtransdu}
A class $\cf S$ of finite relational structures (of the same signature)
is of bounded clique-width, if and only if
$\cf S$ is contained in the image of an MSO transduction of the class of
finite~trees.
\end{theorem}

For a {\em very informal} explanation of the meaning of this statement,
we remark that a tree which is the preimage of the mentioned transduction
gives a hierarchical structure to the clique-width expression in
Definition~\ref{def:cwpoints}.
The arbitrary transduction parameters then determine particular operations
(and labelling) used within the expression, 
and the formula(s) of a final interpretation roughly encode 
Definition~\ref{def:cwpoints} itself.
No copying (``amplification'') is necessary there.

Since the concept of a transduction is transitive, Theorem~\ref{thm:cwtransdu} implies:
\begin{corollary}\label{cor:cwtransdu}
If a class $\cf S$ of order types (of points) is of bounded clique-width, 
then the image of $\cf S$ in an MSO transduction is also of bounded clique-width.
\end{corollary}

\paragraph{Deciding MSO properties. }
Perhaps the most important application of bounded clique-width of point
configurations $P$ could be in faster deciding of MSO-definable properties
(and, in greater generality, of some optimization and counting properties as well,
see examples in \cite{cmr00}) of the order type of~$P$.

\begin{theorem}[Courcelle, Makowsky and Rotics~\cite{cmr00}, via Theorem~\ref{thm:cwtransdu}]
\label{thm:MSOlin}
Consider a class $\cf S$ of finite relational structures of signature
$\sigma$ and of bounded clique-width.
For any MSO sentence $\varphi$ of signature $\sigma$, 
if a structure $S\in\cf S$ is given on the input alongside with a
clique-width expression of bounded width, then we
can decide in linear time whether $S\models\varphi$
(i.e., whether $S$ has the property~$\varphi$).

Furthermore, under the same assumptions for $\cf S$ and for an MSO formula $\varphi(X)$ with
a free set variable $X$, we can find in linear time a maximum-cardinality
set $X$ such that $S\models\varphi(X)$,
and we can enumerate all sets $X$ such that $S\models\varphi(X)$
in time which is linear in the input plus output size.
\end{theorem}

\section{Assorted examples}
\label{sec:assorted}

First, to give readers a better feeling about how big the clique-width of
``nicely looking'' point sets in the plane can be, we show the following:

\begin{theorem}\label{thm:cwbasic2}\apxmark\
Let $P$ be a point configuration, $P_0\subseteq P$ and $d=|P\setminus P_0|$.
\\a) If all points of $P_0$ are collinear, then the clique-width of $P$ is
in $\ca O(d)$.
\\b) Assume the points of $P_0$ are in a strictly convex position.
If $d\leq1$, then the clique-width of $P$ is bounded (by a constant).
On the other hand, there exist examples already with $d=2$ and unbounded
clique-width of $P$.
\end{theorem}

\vspace*{-2ex}\paragraph{Proof outline. }
In case (a), we first create the $d$ points of $P\setminus P_0$, each with
its unique label, and their counter-clockwise order triples.
See Figure~\ref{thm:cwbasic2}(a).
Then we stepwise create the collinear points of $P_0$, ordered from left to right.
During the steps, we add binary labels on $P_0$ between each pair from left
to right, and we also in the right order create the needed order triples
having one point in $P_0$ and two points in $P\setminus P_0$.
At the end, we easily create from the binary labels on $P_0$ the remaining order triples having 
two points in $P_0$ and one in $P\setminus P_0$.

In case (b), if $d=1$, we construct an expression simlarly as in
Proposition~\ref{prop:unaryvs}, but we simultaneously
proceed in two subsequences of the counter-clockwise perimeter of $P_0$,
``opposite'' to each other.
This process allows us to create also the order triples involving the sole
point of $P\setminus P_0$ (in ``the middle'').

In case (b) with $d\geq2$, we present a construction informally shown in
Figure~\ref{thm:cwbasic2}(b).
The underlying idea is to construct collinear triples ``through'' the points
of $P\setminus P_0$, such as the depicted triples $p_0,m,q_0$ and $q_0,n,p_k$.
Since collinear triples are easy to detect within the order type, we can
this way interpret the binary relation between $p_0$ and $p_k$, and
analogously between subsequent $p_1$ and $p_{k+1}$ and so on
(see the green dashed arrows in the picture).
We can also routinely describe in MSO logic the neighbouring pairs of
vertices of a convex hull, see $x$ and $x'$ in \eqref{eq:convhull}.
Consequently, in such a suitably constructed set $P$, we can interpret an
arbitrarily large square grid graph on the points
$p_0,p_1,\ldots,p_k,p_{k+1},\ldots$.
Since the squre grid is a folklore basic example of unbounded
clique-width~\cite{ce12}, we get from Corollary~\ref{cor:cwtransdu}
that the clique-width of such configurations $P$ (with $d=2$) is unbounded.

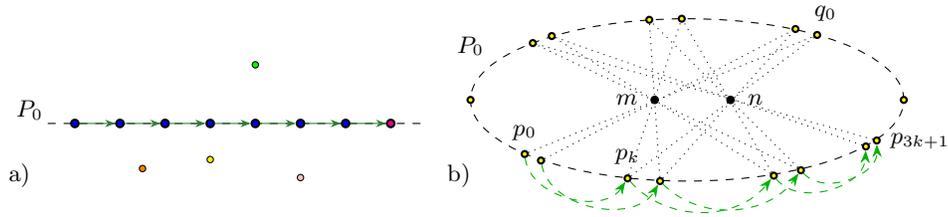
\begin{figure}[t]
\definecolor{dtsfsf}{rgb}{0.8274509803921568,0.1843137254901961,0.1843137254901961}
\definecolor{sexdts}{rgb}{0.1803921568627451,0.49019607843137253,0.19607843137254902}
\definecolor{rvwvcq}{rgb}{0.08235294117647059,0.396078431372549,0.7529411764705882}
\definecolor{dbwrru}{rgb}{0.8588235294117647,0.3803921568627451,0.0784313725490196}
$${\rm a)}\hspace*{-2ex}
\begin{tikzpicture}[scale=0.6, >={Stealth[scale=0.8]}]
\draw[dashed] (-0.6,0) -- (7.7,0);
\tikzstyle{every node}=[draw,thick, shape=circle, inner sep=1pt, color=black, fill=blue]
\node (m1) at (0,0) {}; \node (m2) at (1,0) {}; \node (m3) at (2,0) {}; \node (m4) at (3,0) {};
\node (m5) at (4,0) {}; \node (m6) at (5,0) {}; \node (m7) at (6,0) {}; \node[fill=magenta] (m8) at (7,0) {};
\tikzstyle{every path}=[draw, ->,color=sexdts]
\draw (m1) -- (m2); \draw (m2) -- (m3); \draw (m3) -- (m4); \draw (m4) -- (m5);
\draw (m5) -- (m6); \draw (m6) -- (m7); \draw (m7) -- (m8); 
\draw [color=black,fill=orange] (1.5,-1) circle (2pt);
\draw [color=black,fill=yellow] (3,-0.8) circle (2pt);
\draw [color=black,fill=pink] (5,-1.2) circle (2pt);
\draw [color=black,fill=green] (4,1.3) circle (2pt);
\node[draw=none,fill=none] at (-1,0.3) {$P_0$};
\end{tikzpicture}
\quad{\rm b)}
\begin{tikzpicture}[scale=0.72, >={Stealth[scale=0.8]}]
\path [use as bounding box] (-4,-1.5) rectangle (8,1.4);
\draw[dashed] (0,0) ellipse (40mm and 15mm);
\node[draw=none,fill=none] at (-4,1) {$P_0$};
\tikzstyle{every node}=[draw, shape=circle, inner sep=1pt, color=black, fill=black]
\small
\node[label=left:$m~$] (m1) at (-0.6,0) {};
\node[label=right:$~n$] (m2) at (0.8,0) {};
\tikzstyle{every node}=[draw,thick, shape=circle, inner sep=0.8pt, color=black, fill=yellow]
\node (p1) at (-4,0) {}; \node (p2) at (4,0) {}; 
\node[label=above:$p_0$] (p3) at (-3,-1) {};
\node[label=above:$~q_0\!\!$] (q3) at (2.4,1.2) {}; 
\node[label=above:$p_k$] (pp3) at (-1.1,-1.45) {};
\node (p4) at (-2.7,-1.12) {}; \node (q4) at (2.0,1.31) {}; \node (pp4) at (-0.5,-1.5) {};
\node (q5) at (-0.1,1.5) {}; \node (pp5) at (1.6,-1.4) {};
\node (q6) at (-0.7,1.48) {}; \node (pp6) at (2.1,-1.3) {};
\node (q7) at (-2.5,1.18) {}; \node (pp7) at (3.3,-0.86) {};
\node (q8) at (-2.85,1.05) {}; \node[label=right:$\>p_{3k+1}$] (pp8) at (3.5,-0.75) {};
\tikzstyle{every path}=[dotted, color=black]
\draw (p3) -- (q3); \draw (pp3) -- (q3);
\draw (p4) -- (q4); \draw (pp4) -- (q4);
\draw (pp3) -- (q5); \draw (pp5) -- (q5);
\draw (pp4) -- (q6); \draw (pp6) -- (q6);
\draw (pp5) -- (q7); \draw (pp7) -- (q7);
\draw (pp6) -- (q8); \draw (pp8) -- (q8);
\tikzstyle{every path}=[dashed, color=green!70!black, ->, >={Stealth}]
\draw (p3) to[bend right=70] (pp3);
\draw (p4) to[bend right=70] (pp4);
\draw (pp3) to[bend right=50] (pp5);
\draw (pp4) to[bend right=50] (pp6);
\draw (pp5) to[bend right=70] (pp7);
\draw (pp6) to[bend right=70] (pp8);
\end{tikzpicture}
$$
\caption{Illustrations of the two parts of Theorem~\ref{thm:cwbasic2}.
(a) Labelling for an expression of bounded width.
(b) A sketch of interpreting a large grid within the point configuration.
}
\label{fig:assorted-cw}
\end{figure}

\subsection*{Some NP-hard problems of point configurations}

As already mentioned, perhaps the most interesting computing application
of clique-width of point sets could be in designing algorithms which
run in parameterized polynomial, or even linear, time with respect to the
clique-width as the parameter.
This is especially relevant for problems for which no such algorithms are
believed to exist in general, such as for NP-hard problems.

A parameterized problem has an {\em FPT algorithm} if the algorithm runs in time
$\ca O(f(d)\cdot n^c)$ where $f$ is an arbitrary computable function of the
(fixed) parameter~$d$, and $c$ is a constant. If $c=1$, then we speak about a linear
FPT algorithm (e.g., this is the complete case of Theorem~\ref{thm:MSOlin}).

Since, except graph clique-width, there is no known FPT algorithm (even
approximation one) for finding a clique-width expression of relational
structures of bounded clique-width, we always must assume that an expression
of bounded width is given alongside with the input point configuration.
Notice that for the above presented examples of small clique-width, the
relevant expressions are very natural and easy to come with.
Recall also related Remark~\ref{rem:spacelinear}.

\vspace*{-1ex}%
\paragraph{General position subset. }
This problem asks whether, for a given point set $P$ and integer~$k$,
there exists a subset $Q\subseteq P$ such that no three points of $Q$ are
collinear and $|Q|\geq k$.
This problem is NP-hard and APX-hard by~\cite{DBLP:journals/ijcga/FroeseKNN17}.

\begin{theorem}\label{thm:genposs}
Assume a point set $P$ is given alongside with a clique-width expression
(for $\Omega(P)$) of width $d$.
Then the {\sc General position subset} problem of $P$ is solvable in linear
FPT time with respect to the parameter~$d$.
\end{theorem}
\begin{proof}
We write the MSO formula 
\vspace*{-1ex}$$
\varphi(X)\equiv\> \forall x,y,z \,\in X\, \big[
 x\not=y\not=z\not=x \to \big(\ordtype(x,y,z)\vee\ordtype(y,x,z) \big)\big]
\vspace*{-1ex}$$
to say that no three points in $X$ are collinear,
and then compute using Theorem~\ref{thm:MSOlin} the value
$\max_{\Omega(P)\models\varphi(X)}|X|$ and compare to~$k$.
\qed\end{proof}

A very similar simple approach works also for the NP-hard problem
{\sc Hitting set for induced lines} \cite{DBLP:conf/cocoon/RajgopalAGKM13},
which asks for a minmum-cardinality subset $H\subseteq P$ such that 
the lines between each pair of points of $P$ all contain a point of~$H$.

\vspace*{-1ex}%
\paragraph{Minimum convex partition. }
Consider a given point set $P$ and an integer $k$.
The objective of this problem~\cite{CGchallenge20} is to decide whether 
the convex hull $\prebox{conv}(P)$~of~$P$ can be partitioned into $\leq k$ {convex faces}.
By a {\em convex face} in this situation we mean the convex hull of 
a subset $Q\subseteq P$ which is a convex hole of $P$ (recall \eqref{eq:convhole}).
Note that in our definition $Q$ must be strictly convex, but
we may as well consider the non-strict variant in which 
points of $Q$ are allowed to lie on the boundary of $\prebox{conv}(Q)$ not
in the vertices (and the arguments would be similar).

This problem has been recently claimed
NP-hard~\cite{DBLP:journals/corr/abs-1911-07697}.
Unfortunately, \mbox{inherent} limitations of MSO logic do not allow us to
directly formulate the {\sc Minimum convex partition} as an MSO optimization
problem (one is not allowed~to~quantify set families), but we can handle it
if we take $k$ as an additional parameter.

\begin{theorem}\label{thm:convparti}\apxmark\
Assume a point set $P$ given alongside with a clique-width expression
of width $d$.
The {\sc Minimum convex partition} problem of $P$ into $\leq k$ convex faces
is solvable in linear FPT time with respect to the parameter~$d+k$.
\end{theorem}
\vspace*{-2ex}\paragraph{Proof outline. }
Let $\prebox{convhole}(X)$ denote the MSO formula \eqref{eq:convhole}.
We may now write
\vspace*{-1ex}$$
\exists X_1,\ldots,X_k\> \left[ \bigwedge\nolimits_{1\leq i \leq k}
	\prebox{convhole}(X_i)  ~\wedge~ \prebox{convpartition}(X_1,\ldots,X_k)  \right]
\vspace*{-0ex}$$
where the subformula $\prebox{convpartition}$ checks whether the convex
hulls of the sets $X_i$ partition $\prebox{conv}(P)$.
At this point, we know that each $X_i$ is a convex hole in~$P$, and
we further test for set inclusion and the following two conditions:
\begin{itemize}\vspace*{-1ex}
\item the boundaries of $\prebox{conv}(X_i)$ and $\prebox{conv}(X_j)$
($1\leq i<j\leq k$) do not cross, and
\item every boundary edge of $\prebox{conv}(X_i)$ is, at the same time, a
boundary edge of exactly one of $\prebox{conv}(X_j)$ ($i\not=j$) or of $\prebox{conv}(P)$.
\end{itemize}\vspace*{-1ex}
Both conditions can be, although not easily, stated in MSO over order types.

\paragraph{Terrain guarding. }
Another NP-hard problem formulated on point sets
\cite{DBLP:journals/siamcomp/KingK11} is that of guarding an
$x$-monotone polygonal line $L$ with the given vertex set $P$.
The objective of guarding is to find a minimum-cardinality vertex guard
set $G\subseteq P$ such that every point $\ell$ on $L$ is seen by some 
point $g\in G$ ``from above the terrain'', that is, the straight line segment from $g$
to $\ell$ is never strictly below~$L$.

Note that the point set $P$ (no two points of the same $x$-coordinate)
uniquely determines the terrain~$L$, with the vertices ordered by their
$x$-coordinates as $P=(p_1,p_2,\ldots,p_n)$.
However, the order type $\Omega(P)$ is not (unless we would add an
auxiliary point ``at infinity'' in the $y$-axis direction).
That is why we assume the terrain $L$ given as a relational structure 
consisting of ternary $\Omega(P)$ and the binary successor relation
consisting of the pairs $(p_1,p_{i+1})$ for $1\leq i<n$.

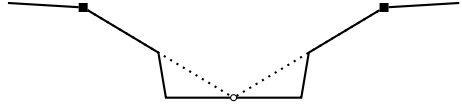
\begin{figure}[t]
$$
\begin{tikzpicture}[yscale=0.6]
\tikzstyle{every node}=[draw, inner sep=1.5pt, color=black, fill=black]
\node (g) at (0,2) {};
\node (gg) at (4,2) {};
\tikzstyle{every node}=[draw=none]
\draw[thick] (-1,2.1) -- (g) -- (1,1) -- (1.1,0) -- (2.9,0) -- (3,1) -- (gg) -- (5,2.1);
\tikzstyle{every node}=[draw, shape=circle, inner sep=0.8pt, color=black, fill=white]
\node (m) at (2,0) {};
\draw[dotted,thick] (g) -- (m) -- (gg);
\end{tikzpicture}
$$
\caption{Guarding a terrain: the two black square vertices guard the whole
terrain, but the bottom horizontal segment is not seen by any single one of them.
To turn this (pair of guards) into a valid segmented terrain guarding solution, we may add
the hollow point as an additional vertex of the terrain.
}
\label{fig:terrhalf}
\end{figure}

There is one further complication in regard of the order type $\Omega(P)$ of
the terrain in this problem:
if, in an instance, some edge of $L$ is seen together by two guards, but no
one sees the full edge, then knowing only $\Omega(P)$ is not sufficient to
verify validity of such a solution (see Figure~\ref{fig:terrhalf}).
That is why we define here the {\sc Segmented terrain guarding} variant in
which every segment of $L$ must be seen by a single vertex guard $g$
and, moreover, there is a dedicated subset $P_1\subseteq P$ such that the
guards are selected from $g\in P_1$.
By a natural subdivision of terrains in the hard instances of terrain guarding
\cite{DBLP:journals/siamcomp/KingK11} we immediately get that also
{\sc Segmented terrain guarding} is NP-hard.

\begin{theorem}\label{thm:terraing}\apxmark\
Assume a polygonal terrain $L$ given alongside with a clique-width expression
of width $d$ (defining both the successor relation and the order type of 
the vertices, cf.~end of Section~\ref{sec:ordertypes}).
The {\sc Segmented terrain guarding} problem of $L$
is solvable in linear FPT time with respect to the parameter~$d$.
\end{theorem}
\vspace*{-2ex}\paragraph{Proof outline. }
We show a formula $\prebox{seguard}(X)$ stating that every segment of
the terrain $L$ is seen by one point of $X$.
There we verify that, for every successive pair of vertices $(p_i,p_{i+1})$ of
$L$, there exists $x\in X$ such that; 
\begin{itemize}\vspace*{-1ex}
\item
the triple $(x,p_i,p_{i+1})$ is oriented counter-clockwise
(for $x$ to see the segment $\overline{p_ip_{i+1}}$ ``from above''),
and 
\item
no ``peak'' $z$ on $L$ between $\overline{p_i,p_{i+1}}$ and $x$ is oriented clockwise
from $(x,p_{i+1})$ (if $z$ is to the left of~$x$) or counter-clockwise from
$(x,p_{i})$ ($z$ to the right of~$x$).
\end{itemize}\vspace*{-1ex}
This suffices since $L$ is $x$-monotone.
Then Theorem~\ref{thm:MSOlin} finishes the argument.
\medskip

We can similarly handle the orthogonal terrain guarding problem
which is also NP-hard~\cite{DBLP:journals/jocg/BonnetG19}.
Another possible extension is to minimize the sum of weighted guards,
using a weighted variant of Theorem~\ref{thm:MSOlin}
(as in~\cite{cmr00}).
However, our approach to terrain guarding cannot be directly extended to
the traditional and more general {\em Art gallery} (guarding) problem
\cite{DBLP:journals/tit/LeeL86}, not even in the adjusted case when each
edge of the polygon is seen by a single vertex guard.
This is due to possible presence of ``blind spots'' in the interior of the
polygon which cannot be determined knowing just the order type $\Omega(P)$
and the boundary edges of the polygon on $P$.
Interested readers may find more in the Appendix.

\paragraph{Polygon visibility graph. }
As we have mentioned the Art gallery problem, we briefly add that people are
also studying problems related to the {\em visibility graph} of a given polygon $Q$.
The visibility graph of $Q$ has the same vertex set as $Q$ and the edges are
those line segments with ends in the vertices of $Q$ which are disjoint from 
the complement of the polygon.
We give the following toolbox:

\begin{theorem}\label{thm:ivisibility}\apxmark\
Assume a polygon $Q$ with vertex set $P$ given as a relational structure
consisting of the order type $\Omega(P)$ and the counter-clockwise
Hamiltonian cycle of the edges of $Q$.
Then the visibility graph of $Q$ has an MSO interpretation in $Q$.
\end{theorem}


\section{Conclusions}
\label{sec:conclu}

We managed to show, in this limited space, only few example applications of
bounding the clique-width in efficient parameterized algorithms for geometric
point problems.
More examples of similar kind could be added but, as a future work, 
we would especially like to investigate possible applications to ``metric'' problems.
Of course, MSO logic of order types cannot express metric properties of a
point set, but it could be possible that in some problems the enumerative part of
Theorem~\ref{thm:MSOlin} provided us with a relatively short list of small
subconfigurations which would then be processed even by brute force,
resulting in a faster algorithm.
For instance, we suggest to investigate in this manner the problem of a 
minimum area triangle on a given point set, which is in general 3SUM-hard 
(that is, not believed to have a subquadratic algorithm).

Another possible extension would be to consider order types in dimension
$3$ (or higher), but then even a strictly convex point set could
easily have unbounded clique-width -- the quaternary relational structures of
such order types just seem to be too complex even in very simple cases.

Lastly, we mention another very natural question; can the clique-width of a
point configuration be at least approximated by 
an FPT algorithm with the width as the fixed parameter?
Such an approximation is possible in the case of graph clique-width \cite{os06,ho08},
thanks to the close relation of graph clique-width to rank-width and to
binary matroids.
Perhaps the natural correspondence of order types to oriented matroids could be of
some help in this research direction.

\subsubsection*{Acknowledgments. }
We would like to thank to Achim Blumensath and Bruno Courcelle for
discussions about the clique-width of relational structures.

\vfill
\bibliography{gtbib,geombib,interpretations}


\newpage

\section*{Appendix}
\appendix

\section{Supplements for Section~\ref{sec:ordertypes}}

\begin{proposition}[\ref{prop:unaryvs}]\label{prop:unaryvsA}
Let $P$ be an arbitrary finite set of points in a strictly convex position.
Then the clique-width of $P$ is bounded by a constant, while the unary
clique-width of $P$ is unbounded.
\end{proposition}

\begin{proof}
We show that $\cwp(P)=4$ by constructing an expression $\cal E$
for $\Omega=\Omega(P)$ using unary labels $1$ and $2$ and a binary label~$3$.
Let the points of $P$ be enumerated as $p_1,p_2,\ldots,p_n$ in
the counter-clockwise order (the starting point does not matter).
We create a starting point $p_1$ of label $1$, and for $k=2,3,\ldots,n=|P|$ 
we iterate this sequence of operations in our expression $\cal E$:
\begin{itemize} \vspace*{-1ex}
\item (w\ref{it:wcreate}) create a new point $p_k$ of label $2$,
and \\(w\ref{it:wdisju}) make the union of previous $\{p_1,\ldots,p_{k-1}\}$
with $\{p_k\}$;
\item (w\ref{it:wpairl}) for all points $a,c$ such that $a$ is of label $1$ and $c$
of label $2$, give the pair $(a,c)$ label $3$;
\item (w\ref{it:wtriple}) for all distinct points $a,b,c$ such that $c$ is
of label $2$ and $(a,b)$ is of label~$3$, add to $\Omega$ the cyclic closure
of $(a,b,c)$;
\item (w\ref{it:wrelab}) relabel all points of label $2$ to label $1$.
\end{itemize}\vspace*{-1ex}

Notice that during every iteration, the only point $c$ of the label $2$ is
$c=p_k$, and after iteration number $k$, the binary label $3$ is given exactly to
pairs $(p_i,p_j)$ such that $1\leq i<j\leq k$.
Hence the operation (w\ref{it:wtriple}) adds triples $(p_i,p_j,p_k)$ and
their cyclic closure, exactly when $1\leq i<j<k$, which all indeed belong to
$\Omega(P)$.
On the other hand, in every counter-clockwise triple $(p_i,p_j,p_k)$ of $P$ we may
assume $k$ is the largest index, and then $i<j<k$ by our indexing of $P$.
So the pair $(p_i,p_j)$ gets label $3$ in above iteration number $j$, and then the
triple $(p_i,p_j,p_k)$ is added to $\Omega$ in iteration number $k$.
Therefore, the value $\Omega$ of $\cal E$ satisfies $\Omega=\Omega(P)$.

Next, in order to prove that the unary clique-width of $P$ is unbounded,
we show that every expression for $\Omega(P)$ which uses only unary
labels -- see (u\ref{it:labeledge}') -- must use at least $\lceil n/2\rceil$
labels where $n=|P|$. 
For a contradiction, let $\cal E$ be an expression with $\ell$ unary labels
whose value is $\Omega(P)$, and assume $n\geq 2\ell+1$.
Imagine the last application of the union operation (u\ref{it:uunion}),
which makes the union of the values of subexpressions $\ca E_1$ and $\ca E_2$.
Since $\ca E_1$ and $\ca E_2$ together make all $n$ points of $P$, 
one of them, say $\ca E_1$, makes at least $\ell+1$ of the points.
By the pigeon-hole principle, some two distinct points $a,b$ created in 
$\ca E_1$ end up with the same label in the value of $\ca E_1$.
(We remark that $a$ and $b$ may be created with different labels, and
possibly relabelled several times, but we speak about the final label they
get within $\ca E_1$.)

Let $c$ be any point created by $\ca E_2$.
Among the two triples $(a,b,c)$ and $(b,a,c)$, exactly one is
counter-clockwise in $P$, say $(a,b,c)$.
The triple $(a,b,c)$ can be created only after the union of $\ca E_1$ and
$\ca E_2$.
However, whenever an application of the operation (u\ref{it:labeledge}')
creates the triple $(a,b,c)$, since the labels of $a$ and $b$ are already
the same and must stay the same, this application adds also the triple
$(b,a,c)$, which is in a contradiction to $\Omega(P)$ being the value of $\ca E$.
\qed\end{proof}

\section{Supplements for Section~\ref{sec:MSO}}

\paragraph{The concept of transductions. }
We provide formal details on the concept of an MSO transduction, which were
skipped for simplicity in the main paper.

Consider relational signatures $\sigma=\{R_1,\dots,R_q\}$ and $\tau=\{R'_1,\dots,R'_{t}\}$.
Recall that a (simple) {MSO interpretation of $\tau$-structures in
$\sigma$-structures} is a $t$-tuple of MSO formulas
$\Psi=(\psi_i: 1\leq i\leq t)$ of the signature $\sigma$,
where the number of free variables of $\psi_i$ equals the arity $a_i$ of $R'_i$.

A {\em basic MSO transduction} $\delta_0$ of a relational $\sigma$-structure~$S$
is a $(t+2)$-tuple $(\chi,\nu,\psi_1,\ldots,\psi_t)$ of MSO formulas of the
signature $\sigma$, where $\chi$ is nullary, $\nu$ is unary and $\psi_i$ are
as in an MSO interpretation $\Psi$ above.
The outcome of $\delta_0(S)$ is undefined (empty) if $S\not\models\chi$, and otherwise,
$\delta_0$ maps the structure $S$ on the ground set $U$ into a single $\tau$-structure 
$T$ on the ground set $U'=\{v\in U\>|~S\models\nu(v)\}$.
The relational symbols of $\tau$ are interpreted on $U'$ as follows;
for each $1\leq i\leq t$, we have $(x_1,\ldots,x_{a_i})\in {R'_i}^{T}$ $\iff$
$S\models \psi_i(x_1,\ldots,x_{a_i})$.

The {\em$m$-copy operation} maps a structure $S$ on the ground set $U$
to the relational structure $S^m$ on the ground set $U^m=U\times\{1,\dots,m\}$, 
such that the subset $U\times\{i\}$
for each $i=1,2\dots,m$ induces a copy of the structure~$S$
(there are no tuples between distinct copies).
Additionally, $S^m$ is equipped with a binary relation $\sim$ and 
unary relations $Q_1,\dots,Q_m$ such that;
$(u,i)\sim(v,j)$ for $u,v\in U$ iff $u=v$, and $Q_i=\{(v,i): v\in U\}$.

The {\em$p$-parameter expansion} maps a structure $S$ to the set of all
structures which result by an expansion of $U$ by $p$ unary predicates
(as arbitrary labels).

Altogether, a many-valued map $\tau$ is an {\em MSO transduction} if it is
$\delta=\delta_0\circ\gamma\circ\varepsilon$ where $\delta_0$ is a basic
MSO transduction, $\gamma$ is a $m$-copy operation for some $m$,
and $\varepsilon$ is a $p$-parameter expansion for some~$p$.

We remark, once again, that the result of a transduction $\delta$ of one structure is
generally a set of structures, due to the involved $p$-parameter expansion.
For a class $\mathcal{C}$ of structures, the result of a {\em transduction $\delta$ of
the class $\mathcal{C}$} is the union of the particular transduction results,
precisely, $\delta(\mathcal{C}):=\bigcup_{S\in\mathcal{C}}\delta(S)$.

\paragraph{Theorem~\ref{thm:cwtransdu} and Definition~\ref{def:cwpoints}. }
We also add more details about the characterization of clique-width in
Theorem~\ref{thm:cwtransdu} and the related aspects of
Definition~\ref{def:cwpoints}.

Blumensath~and Courcelle in~{\cite[Proposition~27]{DBLP:journals/iandc/BlumensathC06}}
claim, in particular, logical equivalence of the following two properties of a
class $\ca C$ of relational structures:
\begin{itemize}\vspace*{-1ex}
\item \cite[Proposition~27\,(iii)]{DBLP:journals/iandc/BlumensathC06},
$\ca C$ is the image of an MSO transduction of a regular subclass of the
class of finite trees (implicitly directed and labelled),
\item \cite[Proposition~27\,(ii)]{DBLP:journals/iandc/BlumensathC06},
$\ca C$ is the set of values of expressions consisting of the disjoint union
and of all quantifier-free operations over a fixed reference signature
(wrt.~$\ca C$).
\end{itemize}\vspace*{-1ex}
The reference signature in the latter point is actually a set of multi-ary
labels used in the expressions whose values form~$\ca C$
(hence the direct correspondence with expressions and their width
in Definition~\ref{def:cwpoints}).

To be formally precise with Theorem~\ref{thm:cwtransdu} and our Corollary~\ref{cor:cwtransdu},
we hence need to formulate the following claim relating the latter point
back to our Definition~\ref{def:cwpoints}:

\begin{lemma}
A class $\cf S$ (of signature~$\sigma$) of order types of point configurations
is of bounded clique-width, if and only if
there exists a reference signature $\tau\supseteq\sigma$, such that
$\cf S$ is the set of values (restricted to signature~$\sigma$) of expressions 
consisting of the disjoint union and of all quantifier-free operations over~$\tau$.
\end{lemma}

\begin{proof}
The forward implication is trivial since the operations in
Definition~\ref{def:cwpoints} are quantifier-free.

For the backward implication, the following is a key idea:
since only quantifier-free operations are used, the decision on an
order-type triple of a member $R\in\cf S$, say triple $(a,b,c)$,
depends only on the labels induced on the subset $\{a,b,c\}$ of the ground
set of~$R$.
This finding inductively extends to all operations in the respective
expression $\ca E$ valued~$R$, which lead to a decision on the triple $(a,b,c)$.
We hence focus on a ``subexpression'' $\ca E_0$ of $\ca E$ which consists
only of the operations having domain in $\{a,b,c\}$.

Consider the last disjoint union operation in $\ca E_0$ which, up to
symmetry, made a disjoint union of substructures on $\{a,b\}$ and $\{c\}$.
The final decision on $(a,b,c)$ depends only on the labels on $\{a,b\}$ and
on $\{c\}$ just before the union operation (further relabellings cannot
bring new piece of information).
Hence we can encode the information leading to a decision on $(a,b,c)$
in a unary label on $c$ and a binary label on $(a,b)$ (or symmetrically on $(b,a)$).
Recursively, we argue that binary labels on $(a,b)$ depend only on the unary
labels of $a$ and~$b$.
Therefore, operations as in Definition~\ref{def:cwpoints} are sufficient to
construct the considered order-type structure~$R$ under suitable signature.
\qed\end{proof}

\section{Supplements for Section~\ref{sec:assorted}}

\begin{theorem}[\ref{thm:cwbasic2}]\label{thm:cwbasic2A}
Let $P$ be a point configuration, $P_0\subseteq P$ and $d=|P\setminus P_0|$.
\\a) If all points of $P_0$ are collinear, then the clique-width of $P$ is
in $\ca O(d)$.
\\b) Assume the points of $P_0$ are in a strictly convex position.
If $d\leq1$, then the clique-width of $P$ is bounded (by a constant).
On the other hand, there exist examples already with $d=2$ and unbounded
clique-width of $P$.
\end{theorem}

\begin{proof}
(a)
We show that $\cwp(P)=d+8$ by constructing an expression $\ca E_a$
for $\Omega_a=\Omega(P)$ using $d+2$ unary labels and three binary labels.
See Figure~\ref{thm:cwbasic2}(a).
We choose a direction of the line $\ell_0$ through $P_0$ and enumerate the points
of $P_0$ in the direction of $\ell_0$ as $p_1,p_2,\ldots,p_n$,
and the points of $P\setminus P_0$ as $q_1,\ldots,q_d$.

As for the expression $\ca E_a$, we first create independently all
points $q_1,\ldots,q_d$, each with its own unique label.
Then we use auxiliary two binary labels to add the needed order triples for
them, based on the unique unary labels.

In the second stage, we create the points $p_1,p_2,\ldots,p_n$, one by one,
at step $i$ assuming that $p_1,\ldots,p_{i-1}$ are labelled $1$
and new $p_i$ gets label $2$.
At this step we give binary label $3$ to all pairs which are of labels
$1$ and $2$ in this order (that is, to the pairs $(p_j,p_i)$ for all $j<i$),
and we also add to $\Omega_a$ all needed order triples involving $p_i$
and two points of $P\setminus P_0$ (since we again have unique labels for such triples).

In the third final stage,
we add the order triples $(q_k,p_j,p_i)$, such that $(p_j,p_i)$ is of the
binary label $3$ and $q_k$ (which still holds its unique label) is in the
counter-clockwise orientation from the line $\ell_0$.
We analogously add the triples $(q_k,p_i,p_j)$ if $q_k$ is clockwise from $\ell_0$.

\medskip(b) Case of $d=1$.
We enumerate the points of $P_0$ in the clockwise orientation as
$p_1,p_2,\ldots,p_n$. Let $P\setminus P_0=\{m\}$.
Let $k$ be the least index such that $(p_i,m,p_k)$ is counter-clockwise.
We denote by $q_1=p_i$, $q_2=p_{i+1}$, $\ldots$
We now use the same construction of an expression as in the proof of
Proposition~\ref{prop:unaryvs}, but now simultaneously applied to
$p_1,p_2,\ldots$ and to $q_1,q_2,\ldots$ such that we always process
together $p_i$ and $q_j$ of the least index such that $(p_i,m,q_j)$ 
stays counter-clockwise.
This results in an expression of width $8$ which correctly adds also the
order triples involving the ``middle'' point $m$.

\medskip(b) Case of $d=2$.
We give a construction, sketched in Figure~\ref{thm:cwbasic2}(b),
of such an order type of unbounded clique-width.
Consider the points of $P\setminus P_0=\{m,n\}$ and the points of $P_0$
named in the counter-clockwise direction $p_0,p_1,\ldots,q_0,q_1,\ldots$
We choose an arbitrary $k$ and place the points of $P_0$ on the convex hull
such that we have collinear triples $(p_i,m,q_i)$ and $(p_{i+k},n,q_i)$
for $i=0,1,2,\ldots,k^2-1$.
No other triples in the construction are collinear.

Within the order type $\Omega(P)$, we can uniquely identify with MSO the
points $m,n$ -- as those two belonging to the convex hull of the remaining points.
We can then write formulas $\prebox{collin}(x,y,z)\equiv \neg\ordtype(x,y,z)
 \wedge\neg\ordtype(y,x,z)$, and
$\gamma_k(x,y)\equiv \exists z \prebox{collin}(x,m,z)\wedge
\prebox{collin}(y,n,z)$.
Then $\gamma_k$ relates $p_i$ to $p_j$ (i.e.,
$\Omega(P)\models\gamma_k(p_i,p_j)\vee\gamma_k(p_j,p_i)$\,) if and only if
$j=i+k$ or vice versa.
We can also in MSO describe the successor relation on the boundary of the convex
hull of $P_0$, that is the pairs $(p_i,p_{i+1})$, cf.~\eqref{eq:convhull}.
In this way we get an MSO interpretation of the $k\times k$ square grid
graph in $\Omega(P)$.

Consequently, in such a suitably constructed set $P$, we can interpret an
arbitrarily large square grid graph on the vertex set
$\{p_0,p_1,\ldots,p_{k^2-1}\}$.
Since the squre grid is a folklore basic example of unbounded
clique-width~\cite{ce12}, we get from Corollary~\ref{cor:cwtransdu}
that the clique-width of such configurations $P$ (with $d=2$) is unbounded.
\qed\end{proof}

\medskip

\begin{theorem}[\ref{thm:convparti}]\label{thm:convpartiA}
Assume a point set $P$ given alongside with a clique-width expression of width $d$.
The {\sc Minimum convex partition} problem of $P$ into $\leq k$ convex faces
is solvable in linear FPT time with respect to the parameter~$d+k$.
\end{theorem}

\begin{proof}
Let $\prebox{convhole}(X)$ denote the MSO formula \eqref{eq:convhole}.
We describe the existence of a partition into $k$ convex faces as follows
$$
\exists X_1,\ldots,X_k\> \left[ \bigwedge\nolimits_{1\leq i \leq k}
	\prebox{convhole}(X_i)  ~\wedge~ \prebox{convpartition}(X_1,\ldots,X_k)  \right]
$$
where the subformula $\prebox{convpartition}$ checks whether the convex
hulls of the sets $X_i$ partition $\prebox{conv}(P)$.
For the latter, we will test the following three conditions:
\begin{itemize}\vspace*{-1ex}
\item for $1\leq i,j\leq k$, $i\not=j$, we have $X_i\not\subseteq X_j$,
\item the boundaries of the polygons $\prebox{conv}(X_i)$ and $\prebox{conv}(X_j)$
do not cross (neither in a vertex, nor in the interior
of an edge), and
\item every boundary edge of each $\prebox{conv}(X_i)$ is, at the same time, a
boundary edge of exactly one other of $\prebox{conv}(X_j)$ or of $\prebox{conv}(P)$.
\end{itemize}\vspace*{-1ex}

We hence first need to define in MSO logic when two points $x,y$ of a
strictly convex set $Y\subseteq P$ form a boundary edge of $\prebox{conv}(Y)$ in the
counter-clockwise orientation from $x$ to $y$.
Analogously to \eqref{eq:convhull}, we write this as
$$
\prebox{bdedge}(Y,x,y)\equiv\>
  x,y\in Y \wedge x\not=y \,\wedge 
    \forall z \big[ (z\in Y\wedge z\not=x,y) \to \ordtype(x,y,z) \big]
.$$

Second, we define in MSO the fact that two straight line segments, $xx'$ and
$yy'$ with distinct ends, cross each other in an interior point of both;
\begin{eqnarray*}
\prebox{ecross}(x,x',y,y')\equiv && \big(
  \ordtype(x,y,y') ~\not\!\!\!\longleftrightarrow \ordtype(x',y,y')
	\big) \\ &&\wedge\> \big(
  \ordtype(y,x,x') ~\not\!\!\!\longleftrightarrow \ordtype(y',x,x')
\big)
.\end{eqnarray*}
The rest of the definition of $\prebox{convpartition}$ 
is a routine combination of the presented MSO formulas.

Finally, we finish by an application of the decision version of Theorem~\ref{thm:MSOlin}.
\qed\end{proof}

\medskip

\begin{theorem}[\ref{thm:terraing}]\label{thm:terraingA}
Assume a polygonal terrain $L$ given alongside with a clique-width expression
of width $d$ (defining both the successor relation and the order type of 
the vertices, cf.~end of Section~\ref{sec:ordertypes}).
The {\sc Segmented terrain guarding} problem of $L$
is solvable in linear FPT time with respect to the parameter~$d$.
\end{theorem}

\begin{proof}
Formally, we view $L$ as a relational structure on the ground set $P$,
over a vocabulary with three relations;
ternary $\ordtype\subseteq P^3$, binary $\prebox{terredge}\subseteq P^2$
and unary $\prebox{canguard}\subseteq P$.
As before, $\ordtype=\Omega(P)$ is interpreted as the order type of $P$,
and $\prebox{terredge}$ is interpreted as the successive pairs in
the left-to-right ordering of points of $P$ on~$L$.
Let $\prebox{terredge}^*$ denote the transitive closure of the successor
relation $\prebox{terredge}$, which is easily expressed in MSO logic.
Lastly, unary $\prebox{canguard}$ is interpreted as the given subset
$P_1\subseteq P$ of allowable guard vertices.

We construct a formula $\prebox{seguard}(X)$ stating that every segment
$\overline{yy'}$ of
the terrain $L$ is seen by a single vertex guard $x$ of $X$, as follows
$$\> \forall y,y' \big[ \prebox{terredge}(y,y') \to
 \exists x\in X \big( \ordtype(x,y,y') \wedge \prebox{nopeak}(x,y,y') \big)\big]. $$
The part $\ordtype(x,y,y')$ certifies that the segment $\overline{yy'}$ of
the polygonal line $L$ is (potentially) seen ``from above'' by a guard at~$x$.
The formula $\prebox{nopeak}(x,y,y')$ then states that no part
(``peak''~$z$) of the terrain $L$ between a guard at $x$
and the segment $\overline{yy'}$ blocks the visibility,
written as
\begin{eqnarray*}\forall z \big[
 \big(\prebox{terredge}(y,y')\wedge
   \prebox{terredge}\nolimits^*(y',z)\wedge\prebox{terredge}\nolimits^*(z,x)\big) \to
 \neg\ordtype(y',x,z) \big]
\\ \wedge
\> \forall z \big[
 \big(\prebox{terredge}(y,y')\wedge
   \prebox{terredge}\nolimits^*(x,z)\wedge\prebox{terredge}\nolimits^*(z,y)\big) \to
 \neg\ordtype(x,y,z) \big].
\end{eqnarray*}

We then formulate $\prebox{seguard}(X) \wedge \forall x\in X
 \prebox{canguard}(x)$ to say that a guard set $X$ is a valid solution.
The optimization version (finding min-cardinaluty~$X$)
of Theorem~\ref{thm:MSOlin} finishes the proof.
\qed\end{proof}

\medskip

\paragraph*{Note on the Art gallery problem. }
We also briefly explain, why our approach to solving terrain guarding
in Theorem~\ref{thm:terraing} cannot be directly extended to
the traditional and more general {\em Art gallery} (guarding) problem
\cite{DBLP:journals/tit/LeeL86}, not even in the adjusted case when each
edge of the polygon is seen by a single vertex guard.

When ``guarding a terrain'', the definition requires the guard set to see
all points of the polygonal line defining that terrain
(and then it comes for free that the guards also see every point ``above''
the terrain).
On the other hand, the usual definition of the Art gallery problem requires
the guards to see every interior point of the polygon in addition to all its
boundary points.
This aspect is important, since observing all boundary points does not
exclude possible presence of ``blind spots'' in the interior of the polygon.

With a simple example in Figure~\ref{fig:Artgbad}, 
we show that, even in our adjusted setting in which every boundary edge must
be seen by a single guard, we cannot determine whether the interior of the
polygon is guarded knowing just the order type $\Omega(P)$.

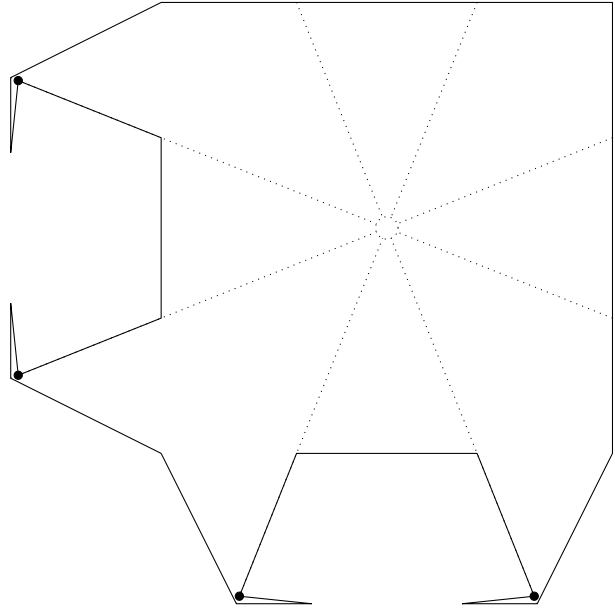
\begin{figure}[t]
$$
\begin{tikzpicture}[scale=1]
\tikzstyle{every node}=[draw, shape=circle, inner sep=1.1pt, color=black, fill=white]
\draw (0,0) -- (1,-2) -- (2,-2) -- (1.04,-1.9) -- (1.8,0) -- (4.2,0) 
    -- (4.96,-1.9) -- (4,-2) -- (5,-2) -- (6,0)
    -- (6,6) -- (0,6) -- (-2,5) -- (-2,4) -- (-1.9,4.96)
    -- (0,4.2) -- (0,1.8) -- (-1.9,1.04) -- (-2,2) -- (-2,1) -- (0,0);
\draw[dotted] (1,-2) -- (4.2,6);
\draw[dotted] (5,-2) -- (1.8,6);
\draw[dotted] (-2,1) -- (6,4.2);
\draw[dotted] (-2,5) -- (6,1.8);
\node[dotted, inner sep=3pt] at (3,3) {};
\node[fill=black] at (1.04,-1.9) {};
\node[fill=black] at (4.96,-1.9) {};
\node[fill=black] at (-1.9,1.04) {};
\node[fill=black] at (-1.9,4.96) {};
\end{tikzpicture}
$$
\caption{Guarding an art gallery: the four black vertex guards see all edges
of the depicted polygon (each edge seen by a single guard), and $4$ is
the obvious necessary minimum.
However, knowing only the order type of the polygon vertices, we cannot say
whether there is a ``blind spot'' in the middle of the polygon
(the depicted dotted spot) or not.
}
\label{fig:Artgbad}
\end{figure}

\medskip

\begin{theorem}[\ref{thm:ivisibility}]\label{thm:ivisibilityA}
Assume a polygon $Q$ with vertex set $P$ given as a relational structure
consisting of the order type $\Omega(P)$ and the counter-clockwise
Hamiltonian cycle of the edges of $Q$.
Then the visibility graph of $Q$ has an MSO interpretation in $Q$.
\end{theorem}

\begin{proof}
Formally, we again view $Q$ as a relational structure on the ground set~$P$,
over a vocabulary with two relations;
ternary $\ordtype=\Omega(P)$ and binary $\prebox{poledge}\subseteq P^2$.
The relation $\prebox{poledge}$ is interpreted as the edge relation of
the counter-clockwise directed cycle of the polygonal edges of $Q$.
Recall that the visibility graph of $Q$ has the same vertex set as $Q$ 
and the edges are those line segments with ends in the vertices of $Q$ 
which are disjoint from the complement of~$Q$.

Our aim is to define an MSO formula $\varepsilon(x,y)$ such that
$Q\models\varepsilon(x,y)$ if and only if $\{x,y\}$ is an edge of the
visibility graph of $Q$.
For that we recall the MSO formula $\prebox{ecross}(x,x',y,y')$ from the
proof of Theorem~\ref{thm:convpartiA} above.
Defining $\varepsilon$ is then a routine exercise using the following facts:
\begin{itemize}
\item Consider a vertex $x\in P$ such that the incoming polygon edge to $x$
is $(x_1,x)$ and the outgoing one is $(x,x_2)$.
Then for any $\{x,y\}$ to be a visibility edge of $Q$, the point $y$ must
lie in the clockwise orientation (non-strictly) from $(x,x_1)$ to $(x,x_2)$.
Written in MSO; $\exists x_1,x_2\big( \prebox{poledge}(x_1,x)\wedge
 \prebox{poledge}(x,x_2)\wedge \neg\ordtype(x,x_1,y)\wedge
 \neg\ordtype(x,y,x_2)\big)$.
\item 
For $\{x,y\}$ to be a visibility edge of $Q$, no polygon edge of $Q$
may strictly cross the line segment $xy$.
That is; $\forall z_1,z_2\big( \prebox{poledge}(z_1,z_2)
 \to \neg\prebox{ecross}(x,y,z_1,z_2)\big)$.
\item 
Conversely, if the previous two conditions are satisfied for each of $x$ and
$y$, then $\{x,y\}$ is a visibility edge of $Q$.
\qed\end{itemize}
\end{proof}

\end{document}